\documentclass[11pt]{article}

\pdfoutput=1

\usepackage{hyperref}
\usepackage{epsfig}
\usepackage{amsmath}
\usepackage{amsthm}
\usepackage{url}
\usepackage{comment}
\usepackage{fullpage}
\usepackage{thm-restate}
\usepackage{datetime}

\usepackage{color}

\newtheorem{theorem}{Theorem}

\newtheorem{proposition}[theorem]{Proposition}
\newtheorem{obs}[theorem]{Observation}

\begin{document}

\title{The number of non-crossing perfect plane matchings \\
is minimized (almost) only by point sets in convex position}
\author{Andrei Asinowski\thanks{
Institut f\"ur Informatik, Freie Univesit\"at Berlin.
E-mail $\langle$\texttt{asinowski@inf.fu-berlin.de}$\rangle$.
Supported by the ESF EUROCORES programme EuroGIGA, CRP `ComPoSe',
Deutsche Forschungsgemeinschaft (DFG), grant \mbox{FE 340/9-1}.
}
}
\date{\empty}

\maketitle

\begin{abstract}
It is well-known that the number of non-crossing perfect matchings of $2k$ points
in convex position in the plane is $C_k$,  the $k$th Catalan number.
Garc\'ia, Noy, and Tejel proved in 2000 that for any set of $2k$ points in {general} position,
the number of such matchings is \textit{at least} $C_k$.
We show that the equality holds \textit{only} for sets of points in convex position,
and for one exceptional configuration of $6$ points.
\end{abstract}

\subsection*{Introduction, notation, result}

Let $S$ be a set of $n=2k$ points in general position
(no three points lie on the same line)
in the plane.
Under a \textit{perfect matching} of $S$
we understand a geometric perfect matching
of the points of $S$
realized by $k$ non-crossing segments.
The number of perfect matchings of $S$ will be denoted by $\mathsf{pm}(S)$.

In general, $\mathsf{pm}(S)$ depends on (the order type of) $S$.
Only for very special configurations an exact formula is known.
The well-known case is that of points in convex position:
\begin{theorem}[Classic/Folklore/Everybody knows]\label{cat}
If $S$ is a set of $2k$ points in convex position,
then $\mathsf{pm}(S)=C_k=\frac{1}{k+1}\binom{2k}{k}$,
the $k$th Catalan number.
\end{theorem}

There are several results concerning the maximum and minimum possible values of $\mathsf{pm}(S)$
over all sets of size $n$.
For the \textit{maximum} possible value of $\mathsf{pm}(S)$, only asymptotic bounds are known.
The best upper bound up to date is due to 
Sharir and Welzl~\cite{Sharir2006} 
who proved that for any $S$ of size $n$,
we have $\mathsf{pm}(S) = O(10.05^n)$.
For the lower bound, Garc{\'i}a, Noy, and Tejel~\cite{gnt} constructed
a family of examples which implies the bound of $\Omega(3^n n^{O(1)})$;
it was recently improved by Asinowski and Rote~\cite{AR} to $\Omega(3.09^n)$.

As for the \textit{minimum} possible value of $\mathsf{pm}(S)$ for sets of size $n=2k$,
Garc\'ia, Noy, and Tejel~\cite{gnt} showed that
it is attained by sets in convex position,
and thus, by Theorem~\ref{cat}, it is $C_{k}=\Omega(2^n/n^{3/2})$:

\begin{theorem}[Garc\'ia, Noy, and Tejel, 2000 \cite{gnt}]\label{gnt}
For any set $S$ of $n=2k$ points in general position in the plane,
we have $\mathsf{pm}(S) \geq C_k$.
\end{theorem}

However, to the best of our knowledge, the question of 
whether \textit{only} sets in convex position have exactly $C_k$ perfect matchings,
was never studied.
In this note we show that this is \textit{almost} the case:
there exists a unique (up to order type) exception shown in Figure~\ref{fig:ex}:
\begin{figure}[h]
$$\includegraphics[scale=1]{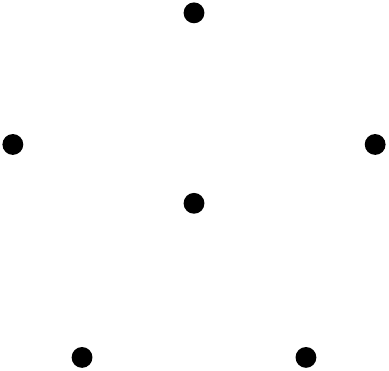}$$
\caption{A set of six points in non-convex position that has five perfect matchings.}
\label{fig:ex}
\end{figure}
\begin{theorem}\label{thm:main}
Let $S$ be a planar set of $2k$ points in general position.
We have $\mathsf{pm}(S) = C_k$ only if
$S$ is in convex position,
or if $k=3$ and $S$ is a set with the order type
as in Figure~\ref{fig:ex}.
\end{theorem}

We recall the recursive definition of Catalan numbers:
$C_0=1$;
and for $k\geq 1$,
\[C_k = \sum_{i=0}^{k-1}C_iC_{k-1-i}.\]
For two distinct points $A$ and $B$, the straight line through $A$ and $B$ will be denoted by $\ell(AB)$.
We say that segment $AB$ \textit{pierces} segment $CD$
if the segments do not cross, but the line $\ell(AB)$ crosses $CD$.

\subsection*{Discussion}
First we recall, for the sake of completeness,
the proof of Theorem~\ref{gnt} by Garc\'ia, Noy, and Tejel.

\begin{proof}
For $k=0$ and $k=1$ the claim is trivial/clear.

Let $k\geq 2$. Refer to Figure~\ref{fig:gnt}.
Let $A_1$ be any point of $S$ that lies on the boundary of $\mathrm{conv}(S)$.
Label other points of $S$ by $A_2, A_3, \dots, A_n$
according to the clockwise polar order with respect to $A_1$
(so that $A_2$ is the immediate successor and $A_n$ is the immediate predecessor of $A_1$
on the boundary of $\mathrm{conv}(S)$).
For $i=0, 1, \dots, k-1$ we bound the number of perfect matchings in which $A_1$ is connected to $A_{2i+2}$, as follows.
The line $\ell(A_1A_{2i+2})$ splits $S \setminus\{A_1, A_{2i+2}\}$ into two subsets
of sizes $2i$ and $n-2-2i$.
Therefore, if we start constructing a perfect matching by choosing the segment $A_1A_{2i+2}$ to be its member,
we can complete its construction by choosing arbitrary perfect matchings of these subsets.
By induction, the numbers of inner perfect matchings of these subsets
are (respectively) at least $C_{i}$ and at least $C_{k-1-i}$.
Thus,
the number of perfect matchings of $S$ in which $A_1$ is matched to $A_{2i+2}$
is at least $C_{i}C_{k-1-i}$, and
the total number of perfect matchings of $S$ is at least
\[\sum_{i=0}^{k-1}C_iC_{k-1-i} = C_k , \]
as claimed.
\end{proof}

\begin{figure}[h]
$$\includegraphics[scale=0.83]{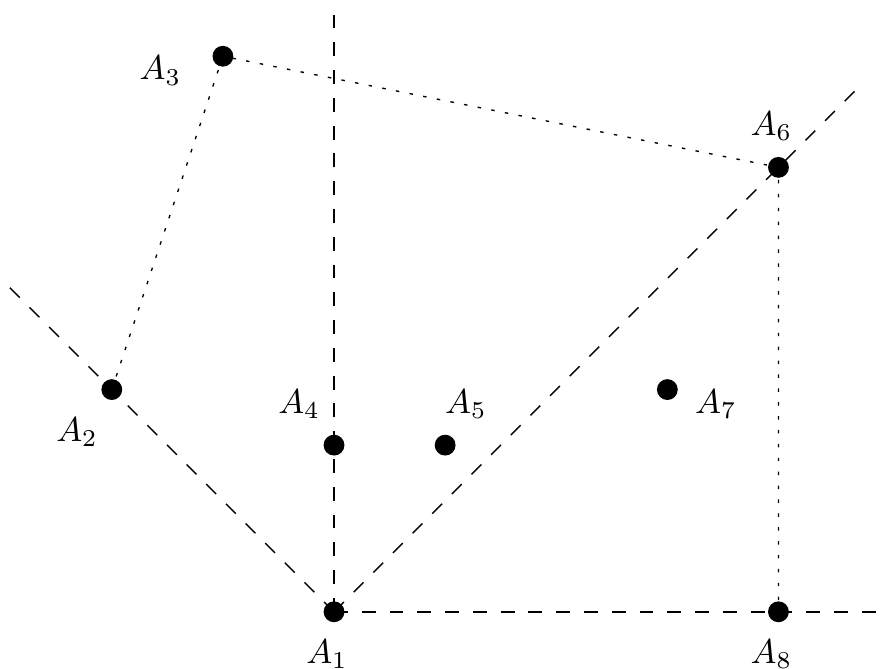}$$
\caption{Illustration to the proof (by Garc\'ia, Noy and Tejel~\cite{gnt}) of Theorem~\ref{gnt}.}
\label{fig:gnt}
\end{figure}

For sets of points in convex position this argument essentially proves Theorem~\ref{cat},
as in this case it counts \textit{all} perfect matchings.
However, for general point sets it is quite rough,
since it does not count
(1) all the perfect matchings in which $A_1$ is matched to $A_j$ with odd $j$
(for example, in Figure~\ref{fig:gnt} we miss perfect matchings that contain the edge $A_1A_5$),
and
(2) all the perfect matchings in which $A_1$ is matched to $A_j$ with even $j$,
and some edges connect pairs of points separated by the line $\ell(A_1A_j)$
(in Figure~\ref{fig:gnt} we miss perfect matchings that contain the edge $A_1A_4$
and some edges that cross $\ell(A_1A_4)$).
Therefore one could expect that we have the equality $\mathsf{pm}(S)=C_k$
for a set of size $2k$ only if it is in convex position
and, maybe, for a limited number of exceptional configurations.
Figure~\ref{fig:ex} shows such a configuration:
it has exactly $5$ ($=C_3$) perfect matchings (the central point can be connected to any other point,
and then a perfect matching can be completed in a unique way).
Thus, in our Theorem~\ref{thm:main} we essentially claim that this is the only exceptional configuration.

\subsection*{Proof of Theorem~\ref{thm:main}}

The main tool will be the following observation.

\begin{obs}\label{obs:pierce}
Let $S$ be a planar set of $2k$ points in general position.
Suppose that $S$ has a perfect matching $M$ in which there are two segments $AB$ and $CD$
such that
one of the endpoints of $AB$ lies on the boundary of $\mathrm{conv}(S)$,
and  $AB$ pierces $CD$ (see Figure~\ref{fig:pierce} for an illustration).
Then $\mathsf{pm}(S) > C_k$.
\end{obs}

\begin{figure}[h]
$$\includegraphics[scale=0.83]{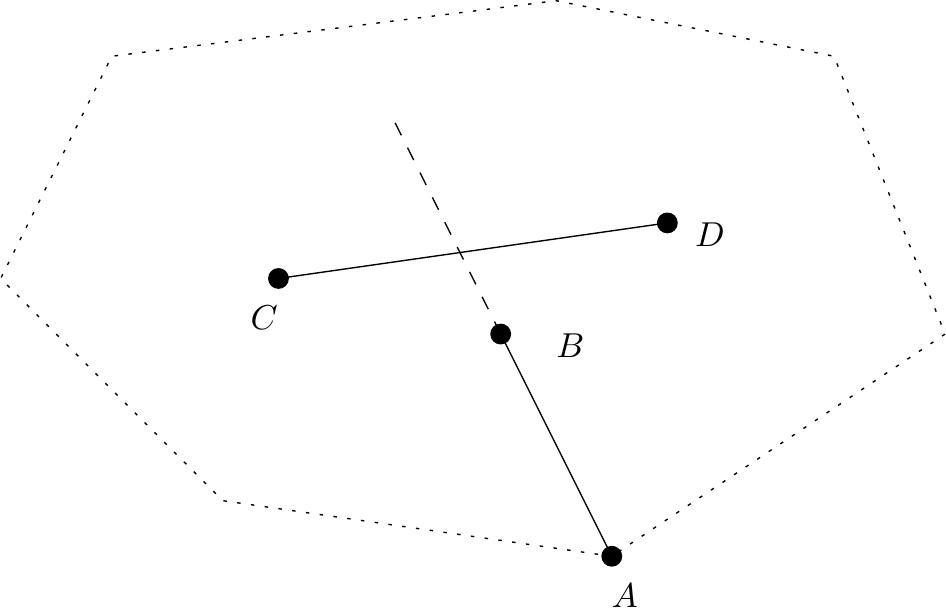}$$
\caption{The \textit{piercing property}.}
\label{fig:pierce}
\end{figure}

\begin{proof}
Set $A_1:=A$ an apply the proof of Theorem~\ref{gnt} for this choice.
It gives $\mathsf{pm}(S)\geq C_k$, but,
as we explained after the proof of Theorem~\ref{gnt}, at least the perfect matching $M$ is not counted.
Therefore we have $\mathsf{pm}(S) > C_k$.
\end{proof}

The property of perfect matchings as in the assumption of Observation~\ref{obs:pierce} 
will be called the \textit{piercing property}.
We shall show that any set of $2k$ points in general position,
except the sets in convex position and
those with order type as in the example from Figure~\ref{fig:ex},
has a perfect matching with the piercing property.
In Propositions~\ref{onepoint} and~\ref{manypoints}
we prove this under assumption that the interior of
$\mathrm{conv}(S)$ contains exactly one or, respectively, several points of $S$.

\begin{proposition}\label{onepoint}
Let $S$ be a set of $n=2k$ points in general position
such that the interior of $\mathrm{conv}(S)$ contains exactly one point of $S$.
Then $S$ has a perfect matching with the piercing property,
unless $S$ has the order type as in Figure~\ref{fig:ex}.
\end{proposition}

\begin{proof}
Let $Q\in S$ be the point that lies in the interior of $\mathrm{conv}(S)$.
Label all other points of $S$
by $A_1, A_2, \dots, A_{n-1}$ according to the clockwise cyclic order
in which they appear on the boundary of $\mathrm{conv}(S)$.
We apply the standard rotating argument on directed lines that pass through $Q$:
we choose one such line and rotate it around $Q$,
keeping track of
the difference $\delta$ between
the number of points of $S$ to the right of the line
and
the number of points of $S$ to its left.
We observe that when the line makes half a turn, $\delta$ changes the sign;
that $\delta$ can only change by $\pm 1$ when the line meets or leaves one of the points of $S$;
and that $\delta$ is even if and only if the line contains one of the points of $S \setminus \{Q\}$.
It follows that for some $j\in\{1, 2, \dots, n-1\}$
we have $\delta=0$ for the line $\ell(QA_j)$.
Thus $\ell(QA_j)$ halves $S$:
there are $k-1$ points of $S$ in each open half-plane bounded by this line.

\smallskip

\textbf{Case 1: $k$ is even.}
Refer to Figure~\ref{fig:cases}(a).
We assume without loss of generality that $\ell(QA_1)$ halves $S$.
Consider the perfect matching $\{A_1Q, A_2A_3, A_4A_5, \dots, A_{n-2}A_{n-1}\}$.
In this matching, $A_1Q$ pierces $A_kA_{k+1}$.
Thus, this matching has the piercing property.

\smallskip

\textbf{Case 2: $k$ is odd.} Here we have two subcases.

\smallskip

\textbf{Subcase 2a: The line $\ell(QA_{j})$ halves $S$ \textit{not for all} $j\in\{1, 2, \dots, n-1\}$.}
Refer to Figure~\ref{fig:cases}(b).
We apply the rotating argument again
and conclude that for some $j_0\in\{1, 2, \dots, n-1\}$
we have $\delta=\pm 2$ for the line $\ell(QA_{j_0})$
(the two signs corresponding to different ways to orient the line).
In other words,
the open half-planes bounded by $\ell(QA_{j_0})$ contain $k-2$ and $k$ points of $S$.
We assume without loss of generality that
$j_0=1$,
and that the open half-planes bounded by $\ell(QA_{1})$ contain respectively the following sets of points:
$\{A_2, A_3, \dots, A_{k-1}\}$ and $\{A_{k}, A_{k+1}, \dots, A_{n-1}\}$.
Consider the perfect matching
$\{A_1Q, A_2A_3, A_4A_5, \dots, A_{n-2}A_{n-1}\}$.
In this matching, $A_1Q$ pierces $A_{k-1}A_{k}$.
Thus, it has the piercing property.

\smallskip

\textbf{Subcase 2b: The line $\ell(QA_{j})$ halves $S$ \textit{for all} $j\in\{1, 2, \dots, n-1\}$.}
Refer to Figure~\ref{fig:cases}(c).
Assume $k \geq 5$.
The segment $A_{k-1}A_{k+2}$ does not cross the segment $A_1Q$:
indeed, the line $\ell(QA_{k-2})$ halves $S$ and thus crosses the segment $A_{n-3}A_{n-2}$;
thus the points $A_{k-1}, A_{k+2}$ on one hand
and the point $A_1$ on the other hand
lie in different open half-planes bounded by $\ell(QA_{k-2})$.
Consider the perfect matching
\[
\left\{A_1Q, \ A_{k-1}A_{k+2}, \ A_{k}A_{k+1}, \ A_iA_{i+1}\colon i\in\{2, 4, \dots, k-5, k-3; k+3, k+5, \dots, n-4, n-2\}  \right\}
\]
In this matching, $A_1Q$ pierces $A_{k-1}A_{k+2}$ (and $A_{k}A_{k+1}$).
Thus, it has the piercing property.

\smallskip

For $k=3$ the argument above (``the segment $A_{k-1}A_{k+2}$ does not cross the segment $A_1Q$'') does not apply.
It is easy to verify by case distinction that only for the order type from Figure~\ref{fig:ex}
there is no perfect matching with the piercing property.
\end{proof}

\begin{figure}[h]
$$\includegraphics[scale=0.83]{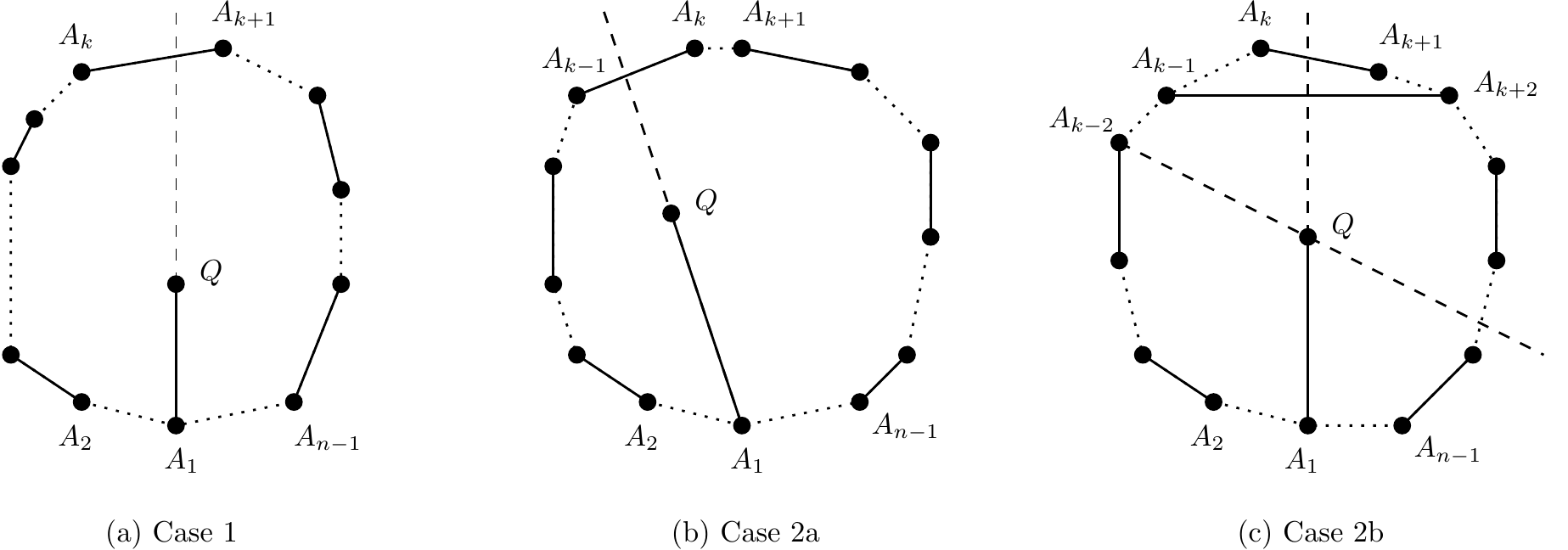}$$
\caption{Illustration of the proof of Proposition~\ref{onepoint}.}
\label{fig:cases}
\end{figure}

\begin{proposition}\label{manypoints}
Let $S$ be a set of $n=2k$ points in general position
such that the interior of $\mathrm{conv}(S)$ contains several points of $S$.
Then $S$ has a matching with the piercing property.
\end{proposition}

\begin{proof}
Label the points of $S$ that lie on the boundary of $\mathrm{conv}(S)$
by $A_1, A_2, \dots, A_{\ell}$ according to the clockwise cyclic
order in which they appear on the boundary.

Let $Q$ and $R$ be two points of $S$ that lie in the interior of $\mathrm{conv}(S)$.
Assume without loss of generality that $\ell(QR)$ crosses the segment $A_1A_2$ so that
$R$ lies on this line between $Q$ and $\ell(QR) \cap A_1A_2$.

Consider the triangle $A_1A_2Q$.
The point $R$ lies in its interior.
Let $R'$ be the point of $S$ in the interior of triangle $A_1A_2Q$
for which the angle $\angle A_1 A_2 R'$ is the smallest.
Denote $T = \ell(A_1 Q) \cap \ell(A_2 R')$.
Due to the choice of $R'$, the interior of triangle $A_1 A_2 T$
does not contain any point of $S$.

\begin{figure}[h]
$$\includegraphics[scale=0.83]{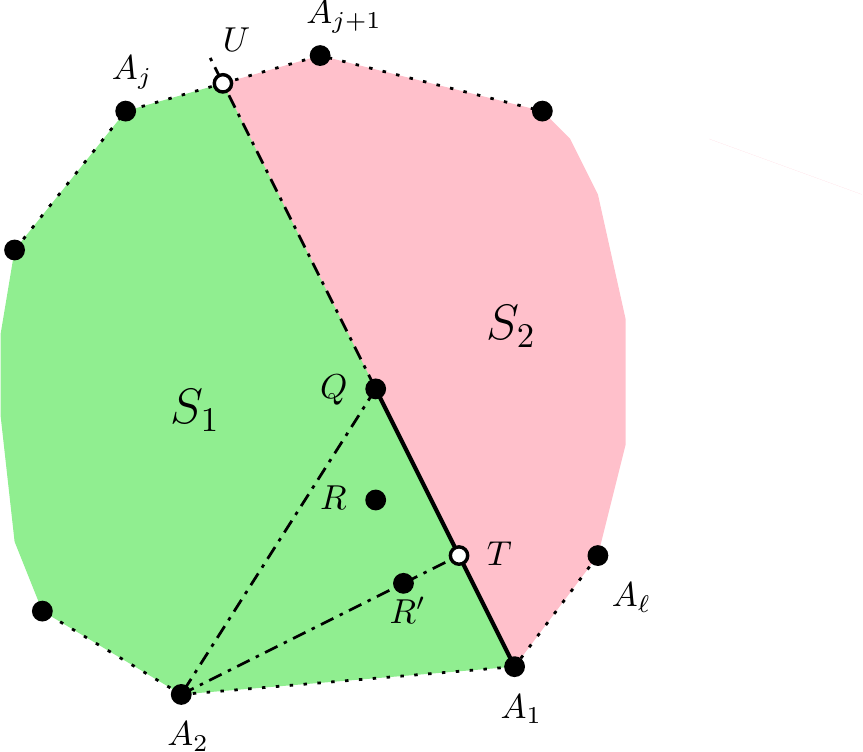}$$
\caption{Illustration of the proof of Proposition~\ref{manypoints}.}
\label{fig:several2}
\end{figure}

The line $\ell(A_1 Q)$ crosses the boundary of $\mathrm{conv}(S)$ twice:
in point $A_1$, and in a point $U$ that belongs to the interior of some segment $A_j A_{j+1}$.
Let $S_1$ and $S_2$ be the sets of points of $S$ that lie respectively in the open half-planes bounded by $\ell (A_1Q)$
(so that $S_1$ contains $A_2$ and $A_j$, and $S_2$ contains $A_{j+1}$ and $A_{\ell}$).
Figure~\ref{fig:several2} illustrates the introduced notation
(the black points belong to $S$, and the white points are reference points that do not belong to $S$).

We start constructing a perfect matching $M$ by taking the segment $A_1 Q$.

If $|S_1|$ and $|S_2|$ are odd, we take the segment $A_jA_{j+1}$ to be a member of $M$,
and then complete constructing $M$ by taking arbitrary perfect matchings of
$S_1\setminus\{A_j\}$ and of $S_2\setminus\{A_{j+1}\}$.
The obtained perfect matching $M$ has the piercing property since $A_1Q$ pierces $A_j A_{j+1}$.

If $|S_1|$ and $|S_2|$ are even, we complete constructing $M$ by taking an arbitrary perfect matching of
$S_2$ and some perfect matching of $S_1$ that contains $A_2 R'$ (this is possible since the interior of triangle $A_1A_2T$
does not contain points of $S$).
The obtained perfect matching $M$ has the piercing property since $A_2 R'$ pierces $A_1 Q$.

Notice that our proof applies as well for the special case when $A_2 = A_j$
(and no other coincidence among the points $A_1, A_2, A_j, A_{j+1}$ is possible).
\end{proof}

Now we can complete the proof of Theorem~\ref{thm:main}.
In Propositions~\ref{onepoint} and~\ref{manypoints} we showed that
any set of $2k$ points in general position has a matching with the piercing property,
unless it is in convex position or has the order type as in the example from Figure~\ref{fig:ex}.
From Observation~\ref{obs:pierce} we know that if $S$ has a matching with the piercing property,
then $\mathsf{pm}(S)> C_k$.

\end{document}